\documentclass{llncs}
\usepackage[utf8]{inputenc}

\usepackage{amsmath,amssymb,amsthm}
\usepackage{float}
\usepackage{graphicx}
\usepackage[vlined,ruled,linesnumbered]{algorithm2e}
\usepackage{amsfonts}
\usepackage{tikz}
\usetikzlibrary{automata}
\usetikzlibrary{arrows.meta}
\usetikzlibrary{decorations.markings}
\usetikzlibrary{decorations.pathreplacing}
\usepackage{url}
\usepackage[inline]{enumitem}
\usepackage{tabularx}
\usepackage{multicol}

\usepackage{times}
\usepackage{xspace}

 \newcommand{\F}{\mathcal{F}}
 
\newcommand{\I}{\mathcal{I}} 
 \renewcommand{\L}{\mathcal{L}}
 
 \renewcommand{\P}{\mathcal{P}}

 \newcommand{\X}{\mathcal{X}}

\newcommand{\LTL}{{\sc ltl}\xspace}
\newcommand{\LTLf}{{\sc ltl}$_f$\xspace}
\newcommand{\ltlf}{{\sc ltl}$_f$\xspace}

\newcommand{\LDLf}{{\sc ldl}$_f$\xspace}
\newcommand{\PLTL}{{\sc pltl}\xspace}
\newcommand{\PLTLf}{{\sc pltl}$_f$\xspace}

\newcommand{\FOL}{{\sc fol}\xspace}
\newcommand{\MSO}{{\sc mso}\xspace}

\newcommand{\NNF}{{\sc nnf}\xspace}
\newcommand{\BNF}{{\sc bnf}\xspace}

\newcommand{\Nat}{{\rm I\kern-.23em N}}

\def\tt{\top}

\def\X{\mathcal{X}}

\def\F{\mathcal{F}}
\def\pltlf{\PLTLf}
\newcommand{\MONA}{$\mathsf{MONA}$\xspace}
\makeatletter
\newcommand{\inlineitem}[1][]{%
\ifnum\enit@type=\tw@
    {\descriptionlabel{#1}}
  \hspace{\labelsep}
\else
  \ifnum\enit@type=\z@
       \refstepcounter{\@listctr}\fi
    \quad\@itemlabel\hspace{\labelsep}
\fi}
\makeatother

\begin{document}

\title{First-Order vs. Second-Order Encodings for \LTLf-to-Automata Translation}

\author{%
Shufang Zhu\inst{1}  \and
Geguang Pu\inst{1}\thanks{Corresponding author} \and
Moshe Y. Vardi\inst{2}
}%

\institute{
East China Normal University, Shanghai, China \and
Rice University, Texas, USA
}

\maketitle

\begin{abstract}

Translating formulas of Linear Temporal Logic~(\LTL) over finite traces, or \LTLf, to symbolic Deterministic Finite Automata~(DFA) plays an important role not only in \LTLf synthesis, but also in synthesis for Safety \LTL formulas. The translation is enabled by using $\mathsf{MONA}$, a powerful tool for symbolic, BDD-based, DFA construction from logic specifications. Recent works used a first-order encoding of \LTLf formulas to translate \LTLf to First Order Logic~(\FOL), which is then fed to \MONA to get the symbolic DFA. This encoding was shown to perform well, but other encodings have not been studied. Specifically, the natural question of whether second-order encoding, which has significantly simpler quantificational structure, can outperform first-order encoding remained open. 

In this paper we address this challenge and study second-order encodings for \LTLf formulas. We first introduce a specific \MSO encoding that captures the semantics of \LTLf in a natural way and prove its correctness. We then explore is a \emph{Compact} \MSO encoding, which benefits from automata-theoretic minimization, thus suggesting a possible practical advantage. To that end, we propose a formalization of symbolic DFA in second-order logic, thus developing a novel connection between BDDs and \MSO. We then show by empirical evaluations that the first-order encoding does perform better than both second-order encodings. The conclusion is that first-order encoding is a better choice than second-order encoding in \LTLf-to-Automata translation.

\end{abstract}

\section{Introduction}\label{sec:intro}
Synthesis from temporal specifications~\cite{PR89a} is a fundamental problem in Artificial Intelligence and Computer Science~\cite{DegVa15}. A popular specification is Linear Temporal Logic (\LTL)~\cite{Pnu77}. The standard approach to solving \LTL synthesis requires, however, determinization of automata on \emph{infinite} words and solving \emph{parity games}, both challenging algorithmic problems~\cite{KV05c}. Thus a major barrier of temporal synthesis has been algorithmic difficulty. One approach to combating this difficulty is to focus on using fragments of \LTL, such as the GR(1) fragment, for which temporal synthesis has lower computational complexity~\cite{BloemGJPPW07}.

A new logic for temporal synthesis, called \LTLf, was proposed recently in~\cite{DV13,DegVa15}. The focus there is not on limiting the syntax of \LTL, but on interpreting it semantically on \emph{finite} traces, rather than \emph{infinite} traces as in~\cite{Pnu77}. Such interpretation allows the executions being arbitrarily long, but not infinite, and is adequate for finite-horizon planning problems. While limiting the semantics to finite traces does not change the computational complexity of temporal synthesis~($\mathsf{2EXPTIME}$), the algorithms for \LTLf are much simpler. The reason is that those algorithms require determinization of automata on \emph{finite} words~(rather than \emph{infinite} words), and solving \emph{reachability} games~(rather than \emph{parity} games)~\cite{DegVa15}.  Another application, as shown in~\cite{ZhuTLPV17}, is that temporal synthesis of \emph{Safety} \LTL formulas, a syntactic fragment of \LTL expressing \emph{safety properties}, can be reduced to reasoning about finite words (see also~\cite{KupfermanVa01,LichtensteinPZ85}). This approach has been implemented in~\cite{ZTLPV17} for \LTLf synthesis and in~\cite{ZhuTLPV17} for synthesis of \emph{Safety} \LTL formulas, and has been shown to outperform existing temporal-synthesis tools such as $\mathsf{Acacia+}$~\cite{BohyBFJR12}. 

The key algorithmic building block in these approaches is a translation of \LTLf to \emph{symbolic} Deterministic Finite Automata~(DFA) \cite{ZhuTLPV17,ZTLPV17}. In fact, translating \LTLf formula to DFA has other algorithmic applications as well. For example, in dealing with safety properties, which are arguably the most used temporal specifications in real-world systems~\cite{KupfermanVa01}. As shown in~\cite{RozierV12}, model checking of safety properties can benefit from using deterministic rather than nondeterminisic automata. Moreover, in runtime verification for safety properties, we need to generate monitors, a type of which are, in essence, deterministic automata~\cite{TabakovRV12}. In~\cite{RozierV12,TabakovRV12}, the translation to deterministic automata is explicit, but symbolic DFAs can be useful also in model checking and monitor generation, because they can be much more compact than explicit DFAs, cf.~\cite{ZTLPV17}.

The method used in~\cite{ZhuTLPV17,ZTLPV17} for the translation of \LTLf to symbolic DFA used an encoding of \LTLf to First-Order Logic~(\FOL) that captures directly the semantics of temporal connectives, and $\mathsf{MONA}$~\cite{KlaEtAl:Mona}, a powerful tool, for symbolic DFA construction from logical specifications. This approach was shown to outperform explicit tools such as $\mathsf{SPOT}$~\cite{spot}, but encodings other than the first-order one have not yet been studied. 
This leads us here to study second-order translations of \LTLf, where we use Monadic Second Order~(\MSO) logic of one successor over finite words~(called {M2L-STR} in~\cite{KlarlundMS00}). Indeed, one possible advantage of using \MSO is the simpler quantificational structure that the second-order encoding requires, which is a sequence of existential monadic second-order quantifiers followed by a single universal first-order quantifier. Moreover, instead of the syntax-driven translation of first-order encoding of \LTLf to \FOL, the second-order encoding employs a semantics-driven translation, which allows more space for optimization. The natural question arises whether second-order encoding outperforms first-order encoding. 

To answer this question, we study here second-order encodings of \LTLf formulas. We start by introducing a specific second-order encoding called \MSO encoding that relies on having a second-order variable for each temporal operator appearing in the \LTLf formula and proving the correctness. Such \MSO encoding captures the semantics of \LTLf in a natural way and is linear in the size of the formula. 
We then introduce a so called \emph{Compact} \MSO encoding, which captures the tight connection between \LTLf and DFAs. We leverage the fact that while the translation from \LTLf to DFA is doubly exponential~\cite{KupfermanVa01}, there is an exponential translation from \emph{Past} \LTLf to DFA (a consequence of~\cite{CKS81,DV13}). Given an \LTLf formula $\phi$, we first construct a DFA that accepts exactly the reverse language satisfying $models(\phi)$ via \emph{Past} \LTLf. We then encode this DFA using second-order logic and ``invert'' it to get a second-order formulation for the original \LTLf formula. Applying this approach directly, however, would yield an \MSO formula with an exponential (in terms of the original \LTLf formula) number of quantified monadic predicates. To get a more compact formulation we can benefit from the fact that the DFA obtained by $\mathsf{MONA}$ from the \emph{Past} \LTLf formula is symbolic, expressed by binary decision diagrams~(BDDs)~\cite{Akers78}. We show how we can obtain a Compact \MSO encoding directly from these BDDs. In addition, we present in this paper the first evaluation of the spectrum of encodings for \LTLf-to-automata from first-order to second-order. 


To perform an empirical evaluation of the comparison between first-order encoding and second-order encoding of \LTLf, we first provide a broad investigation of different optimizations of both encodings. Due to the syntax-driven translation of \FOL encoding, there is limit potential for optimization such that we are only able to apply different normal forms to \LTLf formulas, which are Boolean Normal Form~(\BNF) and Negation Normal Form~(\NNF). The semantics-driven translation of second-order encoding, however, enables more potential for optimization than the \FOL encoding. In particular, we study the following optimizations introduced in~\cite{PSV02,PV03}: in the variable form, where a \emph{Lean} encoding introduces fewer variables than the standard \emph{Full} encoding; and in the constraint form, where a \emph{Sloppy} encoding allows less tight constraints than the standard \emph{Fussy} encoding. The main result of our empirical evaluations is the superiority of the first-order encoding as a way to get $\mathsf{MONA}$ to generate a symbolic DFA, which answers the question of whether second-order outperforms first-order for \LTLf-to-automata translation.

The paper is organized as follows. In Section~\ref{sec:pre} we provide preliminaries and notations. Section~\ref{sec:mso} introduces \MSO encoding and proves the correctness. Section~\ref{sec:omso} describes a more compact second-order encoding, called Compact \MSO encoding and proves the correctness. Empirical evaluation results of different encodings and different optimizations are presented in Section~\ref{sec:exp}. Finally, Section~\ref{sec:con} offers concluding remarks.

\section{Preliminaries}\label{sec:pre}
\subsection{\LTLf Basics}
Linear Temporal Logic over \emph{finite traces} (\ltlf) has the same syntax as \LTL~\cite{DV13}. Given a set $\P$ of propositions, the syntax of \ltlf formulas is as follows:\\
\centerline{
$\phi ::= \top\ |\ \bot\ |\ p\ |\ \neg \phi\ |\ \phi_1\wedge\phi_2\ |\ X\phi\ |\ \phi_1 U \phi_2$}
where $p\in\P$.
We use $\top$ and $\bot$ to denote \textit{true} and \textit{false} respectively. $X$ (Next) and $U$ (Until)
are temporal operators, whose dual operators are $N$ (Weak Next) and $R$ (Release) respectively, defined as 
$N \phi\equiv \neg X\neg \phi$ and $\phi_1 R\phi_2\equiv \neg (\neg\phi_1 U\neg \phi_2)$. The abbreviations $\mbox{(Eventually)}~F\phi\equiv \top U\phi$ and $\mbox{(Globally)}~G\phi\equiv \bot R\phi$ are defined as usual. Finally, we have standard boolean abbreviations, such as $\vee$ (or) and $\rightarrow$ (implies). 

Elements $p\in \P$ are \emph{atom}s. A literal $l$ can be an atom or the negation of an atom. A \textit{trace} $\rho = \rho[0],\rho[1],\ldots$ is a sequence of propositional assignments, where $\rho[x]\in 2^\P$ ($x \geq 0$) is the $x$-th point of $\rho$. Intuitively, $\rho[x]$ is the set of propositions that are $true$ at instant $x$. Additionally, $|\rho|$ represents the length of $\rho$.
The trace $\rho$ is an \textit{infinite} trace if $|\rho| = \infty$ and $\rho\in (2^\P)^{\omega}$; 
otherwise $\rho$ is \textit{finite}, and $\rho\in (2^\P)^{*}$. 
\ltlf formulas are interpreted over finite traces. Given a finite trace $\rho$ and an \ltlf formula 
$\phi$, we inductively define when $\phi$ is $true$ for $\rho$ at point $x$ ($0 \leq x < |\rho|$), written $\rho, x \models \phi$, as follows:

\begin{itemize}
  \item $\rho, x \models\top$ and $\rho, x \not\models\bot$; 
  \item $\rho, x \models p$ iff $p \in \rho[x]$;
  \item $\rho, x \models \neg \phi$ iff $\rho,x \not\models \phi$; 
  \item $\rho, x \models\phi_1 \wedge \phi_2$, iff $\rho,x \models \phi_1$ and $\rho, x \models \phi_2$;
  \item $\rho, x \models X\phi$, iff $x+1 < |\rho|$ and $\rho, x+1 \models \phi$;
  \item $\rho, x \models \phi_1 U \phi_2$, iff there exists $y$ such that $x\leq y < |\rho|$ and $\rho, y\models \phi_2$, and for all $z$, $x \leq z < y$, we have $\rho, z \models \phi_1$.
\end{itemize}

An \ltlf formula $\phi$ is $true$ in $\rho$, denoted by $\rho \models \phi$, when $\rho, 0\models\phi$. Every \LTLf formula can be written in Boolean Normal Form (\BNF) or Negation Normal Form (\NNF)~\cite{RozierV11}. \BNF rewrites the input formula using only $\neg$, $\wedge$, $\vee$, $X$, and $U$. \NNF pushes negations inwards, introducing the dual temporal operators $N$ and $R$, until negation is applied only to atoms.

\subsection {Symbolic DFA and $\mathsf{MONA}$}~\label{SDFA}

We start by defining the concept of \emph{symbolic automaton}~\cite{ZTLPV17}, where a boolean formula is used to represent the transition function of a Deterministic Finite Automaton~(DFA). A symbolic deterministic finite automaton (Symbolic DFA)  $\mathcal{F} = (\P, \X, X_0, \eta, f)$ corresponding to an explicit DFA $\mathcal{D} = (2^\P, S, s_0, \delta, F)$ is defined as follows:
\begin{itemize}
\item $\P$ is the set of atoms; 
\item $\mathcal{X}$ is a set of state variables where $|\X| = \lceil \log_2|S| \rceil$;
\item $X_0 \in 2^\X$ is the initial state corresponding to $s_0$;
\item $\eta : 2^\X\times 2^\P \rightarrow 2^\X$ is a boolean transition function corresponding to $\delta$;
\item $f$ is the acceptance condition expressed as a boolean formula over $\X$ such that $f$ is satisfied by an assignment $X$ iff $X$ corresponds to a final state $s \in F$.
\end{itemize}

We can represent the symbolic transition function $\eta$ by an indexed family $\eta_q : 2^\X \times 2^\P \rightarrow \{0,1\}$ for $x_q \in \X$, which means  that $\eta_q$ can be represented by a binary decision diagram~(BDD)~\cite{Akers78} over $\X \cup \P$. Therefore, the symbolic DFA can be represented by a sequence of BDDs, each of which corresponding to a state variable. 

The $\mathsf{MONA}$ tool~\cite{KlaEtAl:Mona} is an efficient implementation for translating \FOL and \MSO formulas over finite words into minimized symbolic deterministic automata. $\mathsf{MONA}$ represents symbolic deterministic automata by means of \emph{Shared Multi-terminal BDDs} (ShMTBDDs)~\cite{Bryant92,BieKlaRau96}. The symbolic \LTLf synthesis framework of \cite{ZTLPV17} requires standard BDD representation by means of symbolic DFAs as defined above. The transformation from ShMTBDD to BDD is described in~\cite{ZTLPV17}.

\subsection{\FOL Encoding of \LTLf}\label{sec:fol}
First Order Logic (\FOL) encoding of \LTLf translates \LTLf into \FOL over finite linear order with monadic predicates. In this paper, we utilize the \FOL encoding proposed in~\cite{DV13}. We first restrict our interest to \emph{monadic structure}. Consider a finite trace $\rho = \rho[0]\rho[1]\dotsb\rho[e]$, the corresponding \emph{monadic structure} $\I_\rho= (\Delta^\I,<, \cdot^{\mathcal{I}})$ describes $\rho$ as follows. 
$\Delta^\I = \{0, 1, 2,\dotsb, last\}$, where $last = e$ indicating the last point along the trace. The linear order~$<$ is defined over $\Delta^\I$ in the standard way~\cite{KlarlundMS00}. The notation $\cdot^{\mathcal{I}}$ indicates the set of monadic predicates that describe the atoms of $\P$, where the interpretation of each $p \in \P$ is $Q_p = \{x~:~p \in \rho[x]\}$. Intuitively, $Q_p$ is interpreted as the set of positions where $p$ is true in $\rho$. 
In the translation below, 
$\mathsf{fol}(\theta,x)$, where $\theta$ is an \LTLf formula and $x$ is a variable, is an \FOL formula asserting the truth of $\theta$ at point $x$ of the linear order. The translation uses the successor function $+1$, and the variable $last$ that represents the maximal point in the linear order.
\begin{itemize}
\item $\mathsf{fol}(p, x) = (Q_p(x))$
\item $\mathsf{fol}(\neg \phi, x) = (\neg \mathsf{fol}(\phi,x))$
\item $\mathsf{fol}(\phi_1 \wedge \phi_2, x) = (\mathsf{fol}(\phi_1, x) \wedge \mathsf{fol}(\phi_2, x))$
\item $\mathsf{fol}(\phi_1 \vee \phi_2, x) = (\mathsf{fol}(\phi_1, x) \vee \mathsf{fol}(\phi_2, x))$
\item $\mathsf{fol}(X\phi, x) = ((\exists y)((y = x+1) \wedge \mathsf{fol}(\phi, y)))$
\item $\mathsf{fol}(N\phi, x) = ((x = last) \vee ((\exists y)((y = x+1) \wedge \mathsf{fol}(\phi, y))))$
\item $\mathsf{fol}(\phi_1 U \phi_2, x) = ((\exists y)((x\leq y\leq last) \wedge \mathsf{fol}(\phi_2, y) \wedge (\forall z)((x\leq z < y) \rightarrow \mathsf{fol}(\phi_1,z))))$
\item $\mathsf{fol}(\phi_1 R \phi_2, x) = (((\exists y)((x\leq y\leq last) \wedge \mathsf{fol}(\phi_1, y) \wedge (\forall z)((x\leq z \leq y) \rightarrow \mathsf{fol}(\phi_2,z)))) \vee ((\forall z)((x \leq z \leq last) \rightarrow \mathsf{fol}(\phi_2, z))))$
\end{itemize}
For \FOL variables, $\mathsf{MONA}$ provides a built-in operator $+1$ for successor computation. Moreover, we can use built-in procedures in $\mathsf{MONA}$ to represent the variable \emph{last}. 
Given a finite trace $\rho$, we denote  the corresponding finite linear ordered \FOL interpretation  of $\rho$ by  $\mathcal{I}_{\rho}$. The following theorem guarantees the correctness of \FOL encoding of \LTLf.

\begin{theorem}[\cite{Kam68}]\label{thm:fol}
Let $\phi$ be an \LTLf formula and $\rho$ be a finite trace. Then $\rho \models \phi$ iff $\mathcal{I}_{\rho} \models \mathsf{fol}(\phi, 0)$.
\end{theorem}

\section{\MSO Encoding}\label{sec:mso}
First-order encoding was shown to perform well in the context of \LTLf-to-automata translation~\cite{ZhuTLPV17}, but other encodings have not been studied. Specifically, the natural question of whether second-order~(\MSO) outperforms first-order in the same context remained open. \MSO is an extension of \FOL that allows quantification over monadic predicates~\cite{KlarlundMS00}. By applying a semantics-driven translation to \LTLf, we obtain an \MSO encoding that has significantly simpler quantificational structure. 
This encoding essentially captures in \MSO the standard encoding of temporal connectives, cf.~\cite{BCMDH92}.
Intuitively speaking, \MSO encoding deals with \LTLf formula by interpreting every operator with corresponding subformulas following the semantics of the operator. 
We now present \MSO encoding that translates \LTLf formula $\phi$ to \MSO, which is then fed to $\mathsf{MONA}$ to produce a symbolic DFA. 

For an \LTLf formula $\phi$ over a set $\mathcal{P}$ of atoms, let $cl(\phi)$ denote the set of subformulas of $\phi$. We define atomic formulas as atoms $p \in \P$.
For every subformula in $cl(\phi)$ we introduce monadic predicate symbols as follows: for each atomic subformula $p \in \P$, we have a monadic predicate symbol $Q_p$; for each non-atomic subformula $\theta_i \in \{\theta_1, \ldots, \theta_m\}$, we have $Q_{\theta_i}$. 
Intuitively speaking, each monadic predicate indicates the positions where the corresponding subformula is true along the linear order.

Let $\mathsf{mso}(\phi)$ be the translation function that given an \LTLf formula $\phi$ returns a corresponding \MSO formula asserting the truth of $\phi$ at position 0. We define $\mathsf{mso}(\phi)$ as following:
$\mathsf{mso}(\phi) = (\exists Q_{\theta_1}) \cdots (\exists Q_{\theta_m})(Q_{\phi}(0) \wedge (\forall x)(\bigwedge_{i=1}^m \mathsf{t}(\theta_i,x))$, where $x$ indicates the position along the finite linear order. Here $\mathsf{t}(\theta_i,x)$ asserts that the truth of every non-atomic subformula $\theta_i$ of $\phi$ at position $x$ relies on the truth of corresponding subformulas at $x$ such that following the semantics of \LTLf. Therefore, $\mathsf{t}(\theta_i,x)$ is defined as follows:
\begin{itemize}
\item If $\theta_i = (\neg\theta_j)$, then $\mathsf{t}(\theta_i,x) = (Q_{\theta_i}(x) \leftrightarrow \neg Q_{\theta_j}(x))$

\item If $\theta_i = (\theta_j \wedge \theta_k)$, then $\mathsf{t}(\theta_i,x) = (Q_{\theta_i}(x) \leftrightarrow (Q_{\theta_j}(x) \wedge Q_{\theta_k}(x)))$
\item If $\theta_i = (\theta_j \vee \theta_k)$, then $\mathsf{t}(\theta_i,x) = (Q_{\theta_i}(x) \leftrightarrow (Q_{\theta_j}(x) \vee Q_{\theta_k}(x)))$
\item If $\theta_i = (X\theta_j)$, then $\mathsf{t}(\theta_i,x) = (Q_{\theta_i}(x)\leftrightarrow ((x \neq last)\wedge Q_{\theta_j}(x+1)))$
\item If $\theta_i = (N\theta_j)$, then $\mathsf{t}(\theta_i,x) = (Q_{\theta_i}(x)\leftrightarrow ((x = last)\vee Q_{\theta_j}(x+1)))$
\item If $\theta_i = (\theta_j U \theta_k)$, then $\mathsf{t}(\theta_i,x) = (Q_{\theta_i}(x) \leftrightarrow (Q_{\theta_k}(x) \vee ((x \neq last) \wedge Q_{\theta_j}(x) \wedge Q_{\theta_i}(x+1))))$
\item If $\theta_i = (\theta_j R \theta_k$), then $\mathsf{t}(\theta_i,x) = (Q_{\theta_i}(x) \leftrightarrow (Q_{\theta_k}(x) \wedge ((x = last) \vee Q_{\theta_j}(x) \vee Q_{\theta_i}(x+1))))$
\end{itemize}
Consider a finite trace $\rho$, the corresponding interpretation $\I_\rho$ of $\rho$ is defined as in Section~\ref{sec:fol}.
The following theorem asserts the correctness of the \MSO encoding. 

\begin{theorem}\label{thm:msofussy}

Let $\phi$ be an \LTLf formula, $\rho$ be a finite trace. Then $\rho \models \phi$ iff $\I_\rho\models \mathsf{mso}(\phi)$.

\end{theorem}

\begin{proof}

If $\phi$ is a propositional atom $p$, then $\mathsf{mso}(\phi)=Q_p(0)$. It is true that $\rho \models \phi$ iff $\I_\rho \models \mathsf{mso}(\phi)$. If $\phi$ is an nonatomic formula, we prove this theorem in two directions.

Suppose first that $\rho$ satisfies $\phi$. We expand the monadic structure $\I_\rho$ with interpretations for the existentially quantified monadic predicate symbols by setting $Q_{\theta_i}$, the interpretation of subformula $\theta_i$ in $\I_\rho$, as the set collecting all points of $\rho$ satisfying $\theta_i$, that is $Q_{\theta_i}=\{x~:~\rho,x\models\theta_i\}$. We also have $Q_p=\{x~:~\rho,x\models p\}$ and denote the expanded structure by $\I_\rho^{mso}$. By assumption, $Q_\phi(0)$ holds in $\I_\rho^{mso}$.
It remains to prove that $\I_\rho^{mso}\models \forall x.\mathsf{t}(\theta_i,x)$, for each nonatomic subformula $\theta_i\in cl(\phi)$, which we prove via structural induction over $\theta_i$.
\begin{itemize}
\item 
If $\theta_i =(\neg \theta_j$), then $\mathsf{t}(\theta_i,x) = (Q_{\theta_i}(x) \leftrightarrow (\neg Q_{\theta_j}(x)))$. This holds, since $Q_{(\neg\theta_j)}=\{x~:~\rho,x\not\models\theta_j\}$ and $Q_{\theta_j}=\{x~:~\rho,x\models\theta_j\}$.
\item
If $\theta_i = (\theta_j \wedge \theta_k)$, then $\mathsf{t}(\theta_i,x) = (Q_{\theta_i}(x) \leftrightarrow (Q_{\theta_j}(x) \wedge Q_{\theta_k}(x)))$. This holds, since $Q_{(\theta_j\wedge\theta_k)}=\{x~:~\rho,x\models\theta_j\mbox{ and }\rho,x\models\theta_k\}$, $Q_{\theta_j}=\{x~:~\rho,x\models\theta_j\}$ and $Q_{\theta_k}=\{x~:~\rho,x\models\theta_k\}$.
\item
If $\theta_i = (\theta_j \vee \theta_k)$, then $\mathsf{t}(\theta_i,x) = (Q_{\theta_i}(x) \leftrightarrow (Q_{\theta_j}(x) \vee Q_{\theta_k}(x)))$. This holds, since $Q_{(\theta_j\vee\theta_k)}=\{x~:~\rho,x\models\theta_j\mbox{ or }\rho,x\models\theta_k\}$, $Q_{\theta_j}=\{x~:~\rho,x\models\theta_j\}$ and $Q_{\theta_k}=\{x~:~\rho,x\models\theta_k\}$.
\item
If $\theta_i = (X\theta_j)$, then $\mathsf{t}(\theta_i,x) = (Q_{\theta_i}(x)\leftrightarrow ((x \neq last)\wedge Q_{\theta_j}(x+1)))$.
This holds, since $Q_{(X\theta_j)}=\{x~:~\rho,x\models (X\theta_j)\}=\{x~:~ x \neq last\mbox{ and } \rho,x+1\models\theta_j\}$, and $Q_{\theta_j}=\{x~:~\rho,x\models\theta_j\}$.
\item
If $\theta_i = (N\theta_j)$, then $\mathsf{t}(\theta_i,x) = (Q_{\theta_i}(x)\leftrightarrow ((x = last)\vee Q_{\theta_j}(x+1)))$. This holds, since $Q_{(N\theta_j)}=\{x~:~\rho,x\models (N\theta_j)\}=\{x~:~ x=last\mbox{ or } \rho,x+1\models\theta_j\}$, and $Q_{\theta_j}=\{x~:~\rho,x\models\theta_j\}$.

\item If $\theta_i = (\theta_j U \theta_k)$, then $\mathsf{t}(\theta_i,x) = (Q_{\theta_i}(x) \leftrightarrow (Q_{\theta_k}(x) \vee ((x \neq last) \wedge Q_{\theta_j}(x) \wedge Q_{\theta_i}(x+1))))$. This holds, since $Q_{(\theta_j U \theta_k)} = \{x~:~\rho,x \models \theta_j U \theta_k\} = \{x~:~\rho,x \models \theta_k \mbox{ or } x \neq last \mbox{ with } \rho,x \models \theta_j \mbox{ also } \rho,x+1 \models \theta_i \}$, $Q_{\theta_j} = \{x~:~\rho,x \models \theta_j\}$, and $Q_{\theta_k} = \{x~:~\rho,x \models \theta_k\}$;

\item If $\theta_i = (\theta_j R \theta_k)$, then $\mathsf{t}(\theta_i,x) = (Q_{\theta_i}(x) \leftrightarrow (Q_{\theta_k}(x) \wedge ((x = last) \vee Q_{\theta_j}(x) \vee Q_{\theta_i}(x+1))))$. This holds, since $Q_{(\theta_j R \theta_k)} = \{x~:~\rho,x \models \theta_j R \theta_k\} = \{x~:~\rho,x \models \theta_k \mbox{ with } x = last \mbox{ or } \rho,x \models \theta_j \mbox{ or } \rho,x+1 \models \theta_i \}$, $Q_{\theta_j} = \{x~:~\rho,x \models \theta_j\}$, and $Q_{\theta_k} = \{x~:~\rho,x \models \theta_k\}$.
\end{itemize}

Assume now that $\I_\rho\models \mathsf{mso}(\phi)$. This means that there is an 
expansion of $\I_\rho$ with monadic interpretations $Q_{\theta_i}$ for each nonatomic subformula $\theta_i\in cl(\phi)$ such that this expanded structure $\I_\rho^{mso}\models (Q_\phi(0)\wedge ((\forall x)\bigwedge_{i=1}^{m}\mathsf{t}(\theta_i,x)))$. We now prove by induction on $\phi$ that if $x\in Q_{\phi}$, then $\rho,x\models\phi$ such that 
$Q_\phi(0)$ indicates that $\rho,0\models\phi$.
\begin{itemize}
\item 
If $\phi =(\neg \theta_j$), then $\mathsf{t}(\phi,x) = (Q_{\phi}(x) \leftrightarrow (x \notin Q_{\theta_j}))$. Since $\mathsf{t}(\phi)$ holds
at every point $x$ of $\I_\rho^{mso}$, it holds that $x\in Q_{\phi}$ iff $x\not\in Q_{\theta_j}$. It follows by induction that $\rho,x\not\models \theta_j$. Thus, $\rho,x\models\phi$. 
\item
If $\phi = (\theta_j \wedge \theta_k)$, then $\mathsf{t}(\phi,x) = (Q_{\phi}(x) \leftrightarrow (Q_{\theta_j}(x) \wedge Q_{\theta_k}(x)))$. Since $\mathsf{t}(\phi)$ holds at every point $x$ of $\I_\rho^{mso}$, it follows that $x\in Q_{\phi}$ iff $x\in Q_{\theta_j}$ and $x\in Q_{\theta_k}$. It follows by induction that $\rho,x\models\theta_j$ and $\rho,x\models\theta_k$. Thus, $\rho,x\models\phi$.
\item
If $\phi = (\theta_j \vee \theta_k)$, then $\mathsf{t}(\phi,x) = (Q_{\phi}(x) \leftrightarrow (Q_{\theta_j}(x) \vee Q_{\theta_k}(x)))$. Since $\mathsf{t}(\phi)$ holds at every point $x$ of $\I_\rho^{mso}$, it follows that $x\in Q_{\phi}$ iff $x\in Q_{\theta_j}$ or $x\in Q_{\theta_k}$. It follows by induction that $\rho,x\models\theta_j$ or $\rho,x\models\theta_k$. Thus, $\rho,x\models\phi$.
\item
If $\phi = (X\theta_j)$, then $\mathsf{t}(\phi,x) = (Q_{\phi}(x)\leftrightarrow ((x \neq last)\wedge Q_{\theta_j}(x+1)))$.
Since $t(\phi)$ holds at every point $x$ of $\I_\rho^{mso}$, it follows that $x\in Q_{\phi}$ iff $x\neq last$ and $x+1\in Q_{\theta_j}$. It follows by induction that $x\neq last$ and $\rho,x\models\theta_j$. Thus, $\rho,x\models\phi$.
\item
If $\phi = (N\theta_j)$, then $\mathsf{t}(\phi,x) = (Q_{\phi}(x)\leftrightarrow ((x = last)\vee Q_{\theta_j}(x+1)))$. 
Since $\mathsf{t}(\phi)$ holds at every point $x$ of $\I_\rho^{mso}$, it follows that $x\in Q_{\phi}$ iff $x=last$ or $x+1\in Q_{\theta_j}$. It follows by induction that $x=last$ or $\rho,x\models\theta_j$. Thus, $\rho,x\models\phi$.

\item If $\phi = (\theta_j U \theta_k$), then $\mathsf{t}(\phi,x) = (Q_{\phi}(x) \leftrightarrow (Q_{\theta_k}(x) \vee ((x \neq last) \wedge Q_{\theta_j}(x) \wedge Q_{\phi}(x+1))))$. Since $\mathsf{t}(\phi)$ holds at every point $x$ of $\I_\rho^{mso}$, it follows that $x\in Q_{\phi}$ iff $x \in Q_{\theta_k}$ or $x \neq last \mbox{ with } x \in Q_{\theta_j} \mbox{ also }x+1\in Q_{\phi}$. Thus, $\rho,x\models\phi$.

\item If $\phi = (\theta_j R \theta_k$), then $\mathsf{t}(\phi,x) = (Q_{\phi}(x) \leftrightarrow (Q_{\theta_k}(x) \wedge ((x = last) \vee Q_{\theta_j}(x) \vee Q_{\phi}(x+1))))$. Since $\mathsf{t}(\phi)$ holds at every point $x$ of $\I_\rho^{mso}$, it follows that $x\in Q_{\phi}$ iff $x \in Q_{\theta_k}$ with $x = last \mbox{ or } x \in Q_{\theta_j} \mbox{ or }x+1\in Q_{\phi}$. Thus, $\rho,x\models\phi$.

\end{itemize}
\end{proof}

\section{Compact \MSO Encoding}\label{sec:omso}

The \MSO encoding described in Section~\ref{sec:mso} is closely related to the translation of \LTLf to alternating automata \cite{DV13}, with each automaton state corresponding to a monadic predicate. The construction, however, is subject only to syntactic minimization. Can we optimize this encoding using \emph{automata-theoretic minimization?} In fact, $\mathsf{MONA}$ itself applies automata-theoretic minimization. Can we use $\mathsf{MONA}$ to produce a more efficient encoding for $\mathsf{MONA}$?

The key observation is that $\mathsf{MONA}$ can produce a compact symbolic representation of a non-deterministic automaton (NFA) representing a given \LTLf formula, and we can use this symbolic NFA to create a more compact \MSO encoding for \LTLf. This is based on the observation that while the translation from \LTLf to DFA is 2-$\mathsf{EXP}$~\cite{KupfermanVa01}, the translation from \emph{past} \LTLf to DFA is 1-$\mathsf{EXP}$, as explained below. 
We proceed as follows: 
(1) Reverse a given \LTLf formula $\phi$ to \emph{Past} \LTLf formula $\phi^R$; (2) Use \MONA to construct the DFA of $\phi^R$, the reverse of which is an NFA, that accepts exactly the reverse language of the words satisfying $models(\phi)$; (3) Express this symbolic DFA in second-order logic and ``invert'' it to get $\mathcal{D}_{\phi}$, the corresponding DFA of $\phi$.

The crux of this approach, which follows from \cite{CKS81,DV13}, is that
the DFA corresponding to the reverse language of an \LTLf formula $\phi$ of length $n$ has only $2^n$ states. The reverse of this latter DFA is an NFA for $\phi$. We now elaborate on these steps.

\subsection{\LTLf to \pltlf}
Past Linear Temporal Logic over finite traces, i.e. \pltlf, has the same syntax as \PLTL over infinite traces introduced in~\cite{Pnu77}. Given a set of propositions $\mathcal{P}$, the grammar of \pltlf is given by:\\
\centerline{
$\psi ::= \top\ |\ \bot\ |\ p\ |\ \neg \psi\ |\ \psi_1\wedge\psi_2\ |\ Y\psi\ |\ \psi_1 S \psi_2$
}
Given a finite trace $\rho$ and a \pltlf formula  $\psi$, we inductively define when $\psi$ is $true$ for $\rho$ at step $x$ ($0 \leq x < |\rho|$), written by $\rho, x \models \psi$, as follows: 
\begin{itemize}
  \item $\rho, x \models\tt$ and $\rho, x \not\models\bot$;
  \item $\rho, x \models p$ iff $p \in \rho[x]$;
  \item $\rho, x \models \neg \psi$ iff $\rho,x \not\models \psi$;
  \item $\rho, x \models\psi_1 \wedge \psi_2$, iff $\rho,x \models \psi_1$ and $\rho, x \models \psi_2$;
  \item $\rho, x \models Y\psi$, iff $x-1 \geq 0$ and $\rho, x-1 \models \psi$;
  \item $\rho, x \models \psi_1 S \psi_2$, iff there exists $y$ such that $0 \leq y \leq x$ and $\rho, y\models \psi_2$, and for all $z$, $y < z \leq x$, we have $\rho, z \models \psi_1$.
\end{itemize}
A \pltlf formula $\psi$ is $true$ in $\rho$, denoted by $\rho \models \psi$, if and only if $\rho, |\rho|-1\models\psi$. 
To reverse an \LTLf formula $\phi$, we replace each temporal operator in $\phi$ with the corresponding \emph{past} operator of \pltlf thus getting $\phi^R$. $X$(Next) and $U$(Until) correspond to $Y$(Before) and $S$(Since) respectively. 

We define $\rho^R = \rho[|\rho|-1],\rho[|\rho|-2],\ldots,\rho[1],\rho[0]$ to be the reverse of $\rho$. Moreover, given language $\L$, we denote the reverse of $\mathcal{L}$ by $\mathcal{L}^R$ such that $\mathcal{L}^R$ collects all reversed sequences in $\mathcal{L}$. Formally speaking, $\mathcal{L}^R = \{\rho^R ~:~ \rho \in \mathcal{L}\}$.
The following theorem shows that \PLTLf formula $\phi^R$ accepts exactly the reverse language satisfying $\phi$.

\begin{theorem}\label{thm:pltlf}
Let $\mathcal{L}(\phi)$ be the language of \LTLf formula $\phi$ and $\mathcal{L}^R(\phi)$ be the reverse language, then $\mathcal{L}(\phi^R) = \mathcal{L}^R(\phi)$.
\end{theorem}

\begin{proof}
$\mathcal{L}(\phi^R) = \mathcal{L}^R(\phi)$ iff for an arbitrary sequence $\rho \in \mathcal{L}(\phi)$ such that $\rho \models \phi$, it is true that $\rho^R \models \phi^R$. We prove the theorem by the induction over the structure of $\phi$. $last$ is used to denote the last instance such that $last = |\rho|-1$.
\begin{itemize}
\item Basically, if $\phi = p$ is an atom, then $\phi^R = p$, $\rho \models \phi$ iff $p\in\rho[0]$ such that $p \in \rho^R[last]$. Therefore, $\rho^R \models \phi^R$;

\item If $\phi = \neg \phi_1$, then $\phi^R = \neg \phi_1^R$, $\rho \models \neg \phi_1$ iff $\rho \nvDash \phi_1$, such that by induction hypothesis $\rho^R \nvDash \phi_1^R$ holds, therefore $\rho^R \models \phi^R$ is true;

\item If $\phi =\phi_1 \wedge \phi_2$, then $\phi^R = \phi_1^R \wedge \phi_2^R$, $\rho \models \phi$ iff $\rho$ satisfies both $\phi_1$ and $\phi_2$. By induction hypothesis $\rho^R \models \phi_1^R$ and $\rho^R \models \phi_2^R$ hold, therefore $\rho^R \models \phi^R$ is true;

\item If $\phi= X\phi_1$, $\phi^R = Y\phi_1^R$, $\rho \models \phi$ iff suffix $\rho'$ is sequence$\rho[1],\rho[2],\ldots,\rho[last]$ and $\rho' \models \phi_1$. By induction hypothesis, $\rho'^R \models \phi_1^R$ holds, in which case $\rho^R,{last-1} \models \phi_1^R$ is true, therefore $\rho^R \models \phi^R$ holds. 

\item If $\phi = \phi_1 U \phi_2$, $\rho \models \phi$ iff there exists $y$ such that $y~(0\leq y \leq last)$, suffix $\rho' = \rho[y],\rho[y+1],\ldots,\rho[last]$ satisfies $\phi_2$. Also for all $z$ such that $z~(0\leq z< y)$, $\rho'' = \rho[z],\rho[z+1],\ldots,\rho[last]$ satisfies $\phi_1$. By induction hypothesis, $\rho'^R \models \phi_1^R$ and $\rho''^R \models \phi_2^R$ hold, therefore we have $\rho^R,last-y \models \phi_2^R$ and $\forall z.last-y < z \leq last, \rho^R,z \models \phi_1^R$ hold such that $\rho^R \models \phi^R$. The proof is done.
\end{itemize}
\end{proof}

\subsection{\pltlf to DFA}
The DFA construction from \pltlf formulas relies on $\mathsf{MONA}$ as well. Given \PLTLf formula $\psi$, we are able to translate $\psi$ to \FOL formula as input of $\mathsf{MONA}$, which returns the DFA. For \pltlf formula $\psi$ over $\mathcal{P}$, we construct the corresponding \FOL formula with respect to point $x$ by a function $\mathsf{fol_p}(\psi,x)$ asserting the truth of $\psi$ at $x$. Detailed translation of \PLTLf to \FOL is defined below. The translation uses the predecessor function $-1$, and the predicate $last$ referring to the last point along the finite trace.
\begin{itemize}
\item $\mathsf{fol_p}(p, x) = (Q_p(x))$
\item $\mathsf{fol_p}(\neg \psi, x) = (\neg \mathsf{fol_p}(\psi,x))$
\item $\mathsf{fol_p}(\psi_1 \wedge \psi_2, x) = (\mathsf{fol_p}(\psi_1, x) \wedge \mathsf{fol_p}(\psi_2, x))$
\item $\mathsf{fol_p}(Y \psi, x) = ((\exists y)((y = x-1) \wedge (y \geq 0) \wedge \mathsf{fol_p}(\psi, y)))$
\item $\mathsf{fol_p}(\psi_1 S \psi_2, x) = ((\exists y)((0\leq y\leq x) \wedge \mathsf{fol_p}(\psi_2, y) \wedge (\forall z)((y < z \leq x) \rightarrow \mathsf{fol_p}(\psi_1,z))))$
\end{itemize}
Consider a finite trace $\rho$, the corresponding interpretation $\mathcal{I}_{\rho}$ is defined as in Section~\ref{sec:fol}. The following theorem guarantees the correctness of the above translation.

\begin{theorem}\label{thm:pltlf2fol}{\rm \cite{Kam68}}
Let $\psi$ be a \PLTLf formula, $\rho$ be a finite trace. Then $\rho \models \psi$ iff $\mathcal{I}_{\rho} \models \mathsf{fol_p}(\psi, last)$, where $last = |\rho|-1$.
\end{theorem}

\begin{proof}
We prove the theorem by the induction over the structure of $\psi$.
\begin{itemize}
\item Basically, if $\psi = p$ is an atom, $\rho \models \psi$ iff $p\in\rho[last]$. By the definition of $\mathcal{I}$, we have that $last \in Q_p$. Therefore, $\rho \models \psi$ iff $\mathcal{I}_{\rho} \models \mathsf{fol_p}(p, last)$ holds;

\item If $\psi = \neg \psi$, $\rho \models \neg \psi$ iff $\rho \nvDash \psi$. By induction hypothesis it is true that $\mathcal{I}_{\rho} \nvDash \mathsf{fol_p}(\psi, last)$, therefore $\mathcal{I}_{\rho} \models \mathsf{fol_p}(\neg \psi, last)$ holds;

\item If $\psi =\psi_1 \wedge \psi_2$, $\rho \models \psi$ iff $\rho$ satisfies both $\psi_1$ and $\psi_2$. By induction hypothesis, it is true that $\mathcal{I}_{\rho} \models \mathsf{fol_p} (\psi_1,last)$ and $\mathcal{I}_{\rho} \models \mathsf{fol_p}(\psi_2, last)$. Therefore $\mathcal{I}_{\rho} \models \mathsf{fol_p}(\psi_1,last)\wedge \mathsf{fol_p}(\psi_2, last)$ holds;

\item If $\psi= Y\psi_1$, $\rho \models \psi$ iff prefix $\rho'= \rho[0],\rho[1],\ldots,\rho[last-1]$ of $\rho$ satisfies $\rho' \models \psi_1$. Let $\mathcal{I}_{\rho}'$ be the corresponding interpretation of $\rho'$, thus for every atom $p \in \P$, $x\in Q_p'$ iff $x\in Q_p$ where $Q_p'$ is the corresponding monadic predicate of $p$ in $\I_\rho'$. By induction hypothesis it is true that $\mathcal{I}_{\rho}'\models \mathsf{fol_p} (\psi_1, last-1)$, therefore $\mathcal{I}_{\rho}\models \mathsf{fol_p} (Y\psi_1, last)$ holds. 

\item If $\psi = \psi_1 S \psi_2$, $\rho \models \psi$ iff there exists $y$ such that $0\leq y \leq last$ and prefix $\rho' = \rho[0],\rho[1],\ldots,\rho[y]$ of $\rho$ satisfies $\psi_2$ and for all $z$ such that $y< z \leq last$, $\rho'' = \rho[0],\rho[1],\ldots,\rho[z]$ satisfies $\psi_1$. Let $\mathcal{I}_{\rho}'$ and $\mathcal{I}_{\rho}''$ be the corresponding interpretations of $\rho'$ and $\rho''$. Thus for every atom $p \in \P$ it is true that 
$x\in Q_p'$ iff $x \in Q_p$, $x\in Q_p''$ iff $x \in Q_p$, where $Q_p'$ and $Q_p''$ correspond to the monadic predicates of $p$ in $\I_\rho'$ and $\I_\rho''$ respectively. By induction hypothesis it is true that $\mathcal{I}_{\rho}'\models \mathsf{fol_p} (\psi_2, last-y)$ and $\mathcal{I}_{\rho}''\models \mathsf{fol_p} (\psi_1, last-z)$ hold, therefore $\mathcal{I}_{\rho} \models \mathsf{fol_p}(\psi_1 S \psi_2,last)$. 
\end{itemize}
\end{proof}

\subsection{Reversing DFA via Second-Order Logic}
For simplification, from now we use $\psi$ to denote the corresponding \PLTLf formula $\phi^R$ of \LTLf formula~$\phi$. We first describe how BDDs represent a symbolic DFA. Then we introduce the Compact \MSO encoding that inverts the DFA by formulating such BDD representation into a second-order formula.
The connection between BDD representation and second-order encoding is novel, to the best of our knowledge.

As defined in Section~\ref{SDFA}, given a symbolic DFA $\mathcal{F}_{\psi} = (\P, \X, X_0, \eta, f)$ represented by a sequence $\mathcal{B} = \langle B_0, B_1, \ldots, B_{k-1}\rangle$ of BDDs, where there are $k$ variables in $\X$, a run of such DFA on a word $\rho = \rho[0],\rho[1],\ldots,\rho[e-1]$ involves a sequence of states $\xi = X_0, X_1, \ldots, X_e$ of length $e+1$. For the moment if we omit the last state reached on an input of length $e$, we have a sequence of states $\xi' = X_0,X_1, \ldots, X_{e-1}$ of length~$e$. Thus we can think of the run $\xi'$ as a labeling of the positions of the word with states, which is $(\rho[0], X_0), (\rho[1],X_1),\ldots,(\rho[e-1],X_{e-1})$. At each position with given word and state, the transition moving forward involves a computation over every $B_q$ ($0 \leq q \leq k-1$). To perform such computation, take the high branch in every node labeled by variable $v \in \{\X \cup \P\}$ if $v$ is assigned $1$ and the low branch otherwise.

The goal here is to write a formula $\mathsf{Rev}(\F_\psi)$ such that there is an accepting run over $\mathcal{F}_{\psi}$ of a given word $\rho$ iff $\rho^R$ is accepted by $\mathsf{Rev}(\F_\psi)$. To do this, we introduce one second-order variable $V_q$ for each $x_q \in \X$ with $0 \leq q \leq k-1$, and one second-order variable $N_\alpha$ for every nonterminal node $\alpha$ in BDDs, $u$ nonterminal nodes in total. The $V_q$ variables collect the positions where $x_q$ holds, and the $N_\alpha$ variables indicate the positions where the node $\alpha$ is visited, when computing the transition. To collect all transitions moving towards accepting states, we have BDD $B_f' = f(\eta(\X, \P))$. 

Here are some notations. Let $\alpha$ be a nonterminal node, $c$ be a terminal node in $B_q$ such that $c \in \{0,1\}$ and $d \in \{0,1\}$ be the value of $v$. For nonterminal node $\alpha$, we define:\\
\centerline{$\mathsf{Pre}(\alpha) = \{(\beta,v,d)~:~\mbox{there is an edge from } \beta \mbox{ to } \alpha \mbox{ labelled by } v=d\}$}\\
\centerline{$\mathsf{Post}(\alpha) = \{(\beta,v,d)~:~\mbox{there is an edge from } \alpha \mbox{ to } \beta \mbox{ labelled by } v=d\}$}\\
For every terminal node $c$ in BDD $B_q$, we define:\\
\centerline{$\mathsf{PreT}(B_q,c) = \{(\beta,v,d)~:~\mbox{there is an edge from } \beta \mbox{ to } c \mbox{ labelled by } v=d \mbox{ in BDD } B_q\}$}\\
Also, we use $\in^d$ to denote $\in$ when $d=1$ and $\notin$ when $d=0$. For each BDD $B_q$, $\mathsf{root}(B_q)$ indicates the root node of $B_q$.

We use these notations to encode the following statements:

(1) At the \textbf{last} position, state $X_0$ should hold since $\xi'$ is being inverted and $X_0$ is the starting point;\\
\centerline{$\mathsf{Rinit} = (x=last) \rightarrow \left(\bigwedge_{0\leq q \leq k-1, X_0(x_q) = d}x \in^d V_q \right)$;} 

(2) At position $x$, if the current computation is at nonterminal node $\alpha$ labeled by $v$, then (2.a) the current computation must come from a predecessor labeled by $v'$ following the value of $v'$, and (2.b) the next step is moving to the corresponding successor following the value of $v$;\\
\centerline{$\mathsf{node} = \bigwedge_{1\leq \alpha\leq u} (~\mathsf{PreCon} \wedge \mathsf{PostCon}~)$; where}
\begin{align*}
\mathsf{PreCon} &= \left( x \in N_\alpha \rightarrow (\bigvee_{(\beta,v',d)\in \mathsf{Pre}(\alpha)} [x\in N_\beta \wedge x\in^d v' ])\right)  \\
\mathsf{PostCon} &= \left(\bigwedge_{(\beta,v,d) \in \mathsf{Post}(\alpha)} [x \in N_\alpha \wedge x \in^d v \rightarrow x \in N_\beta] \right).
\end{align*}

(3) At position $x$ such that $\textbf{x}>\textbf{0}$, if the current computation node $\alpha$ moves to a terminal node $c$ of $B_q$, then the value of $x_q$ at position $\textbf{x-1}$ is given by the value of $c$. Such computations of all $B_q(0 \leq q \leq k-1)$ finish one transition;
\begin{align*}
\mathsf{Rterminal}= \bigwedge_{0\leq q \leq k-1} \left(\bigwedge_{(\beta,v,d)\in \mathsf{PreT}(B_q,c)} [(x>0 \wedge x \in N_\beta \wedge x \in^d v) \rightarrow (x-1 \in^c V_q)] \right);
\end{align*}

(4) At the \textbf{first} position, the current computation on $B_f'$ has to surely move to terminal 1, therefore terminating the running trace of $\xi$.
\begin{align*}
\mathsf{Racc} = (x=0) \rightarrow \left( \bigvee_{(\beta,v,d) \in \mathsf{PreT}(B_f',1)}[x \in N_\beta \wedge x \in^d v] \right).
\end{align*}

To get all computations over BDDs start from the root at each position, we have\\
\centerline{$\mathsf{roots} = \bigwedge_{0\leq x \leq last} \bigwedge_{0\leq q \leq k-1} x \in \mathsf{root}(B_q)$}
$\mathsf{Rev}(\F_\psi)$ has to take a conjunction of all requirements above such that\\ 
\centerline{
$\mathsf{Rev}(\F_\psi) = (\exists V_0)(\exists V_1) \ldots (\exists V_{k-1})(\exists N_1)(\exists N_2) \ldots (\exists N_u)(\forall x)(\mathsf{Rinit} \wedge \mathsf{node}$}
\centerline{
$\wedge \mathsf{Rterminal} \wedge \mathsf{Racc} \wedge \mathsf{roots}).$}

Therefore, let $\mathsf{Cmso}(\phi)$ be the translation function that given an \LTLf formula $\phi$ returns a corresponding second-order formula applying the Compact \MSO encoding, we define  $\mathsf{Cmso}(\phi) = \mathsf{Rev}(\F_\psi)$ asserting the truth of $\phi$ at position 0, where $\psi$ is the corresponding \PLTLf formula of $\phi$, and $\F_\psi$ is the symbolic DFA of $\psi$. The following theorem asserts the correctness of the Compact \MSO encoding.

\begin{theorem}
The models of formula $\mathsf{Cmso}(\phi)$ are exactly the words satisfying $\phi$.
\end{theorem}
\begin{proof}
 We first have that $\L(\phi) = \L^R(\psi) = \L^R(\F_\psi)$ holds since $\psi$ is the corresponding \PLTLf formula of $\phi$ and $\F_\psi$ collects exactly the words satisfying $\psi$.
 Moreover, $\L(\mathsf{Rev}(\F_\psi)) = \L^R(\F_\psi)$ is true following the construction rules of $\mathsf{Rev}(\F_\psi)$ described above and $\mathsf{Cmso}(\phi) = \mathsf{Rev}(\F_\psi)$. Therefore, $\L(\phi) = \L(\mathsf{Cmso}(\phi))$ holds, in which case the models of formula $\mathsf{Cmso}(\phi)$ are exactly the words satisfying $\phi$.
\end{proof}

Notice that the size of $\mathsf{Cmso}(\phi)$ is in linear on the size of the BDDs, which lowers the logical complexity comparing to the \MSO encoding in Section~\ref{sec:mso}. Moreover, in the Compact \MSO encoding, the number of existential second-order symbols for state variables are nevertheless possibly less than that in \MSO encoding, but new second-order symbols for nonterminal BDD nodes are introduced. BDDs provide a compact representation, in which redundant nodes are reduced. Such advantages allow Compact \MSO encoding to use as few second-order symbols for BDD nodes as possible.


\section{Experimental Evaluation}\label{sec:exp}
We implemented proposed second-order encodings in different parsers for \LTLf formulas using C++. Each parser is able to generate a second-order formula corresponding to the input \LTLf formula, which is then fed to \MONA~\cite{KlaEtAl:Mona} for subsequent symbolic DFA construction. Moreover, we employed \emph{Syft}'s~\cite{ZTLPV17} code to translate \LTLf formula into first-order logic~(\FOL), which adopts the first-order encoding described in Section~\ref{sec:fol}.

\vspace{0.2cm}
\noindent\textbf{\emph{Benchmarks}}
We conducted the comparison of first-order encoding with second-order encoding in the context of \LTLf-to-DFA, thus only satisfiable but not valid formulas are interesting. Therefore, we first ran an \LTLf satisfiability checker on \LTLf formulas and their negations to filter the valid or unsatisfiable formulas. We collected 5690 formulas, which consist of two classes of benchmarks: 765 \LTLf-specific benchmarks, of which 700 are scalable \LTLf pattern formulas from~\cite{CiccioMM16} 
and 65 are randomly conjuncted common \LTLf formulas from~\cite{GiacomoMM14,CiccioM15,PrescherCM14}
; and 4925 \LTL-as-\LTLf formulas from~\cite{RozierV07,RozierV11}, since \LTL formulas share the same syntax as \LTLf.

\vspace{0.2cm}
\noindent\textbf{\emph{Experimental Setup}}
To explore the comparison between first-order and second-order for \LTLf-to-DFA translation, we ran each formula for every encoding on a node within a high performance cluster. These nodes contain 12  processor cores at 2.2 GHz each with 8GB of RAM per core. Time out was set to be 1000 seconds. Cases that cannot generate the DFA within 1000 seconds generally fail even if the time limit is extended, since in these cases, $\mathsf{MONA}$ typically cannot handle the large BDD.

\subsection{Optimizations of Second-Order Encoding}

Before diving into the optimizations of second-order encoding, we first study the potential optimization space of the first-order encoding that translates \LTLf to \FOL. Due to the syntax-driven translation of \FOL encoding, we are only able to apply different normal forms, Boolean Norma Form~(\BNF) and Negation Normal Form~(\NNF). We compared the impact on performance of \FOL encoding with two \LTLf normal forms. It turns out that the normal form does not have a measurable impact on the performance of the first-order encoding. Since \FOL-\BNF encoding performs slightly better than \FOL-\NNF, the best \FOL encoding refers to \FOL-\BNF.

To explore the potential optimization space of the second-order encodings proposed in this paper, we hope to conduct experiments with different optimizations. We name second-order encoding with different optimizations \emph{variations}. We first show optimizations of the \MSO encoding described in Section~\ref{sec:mso}, then describe variations of the Compact \MSO encoding shown in Section~\ref{sec:omso} in the following.

The basic \MSO encoding defined in Section~\ref{sec:mso} translates \LTLf to \MSO in a natural way, in the sense that introducing a second-order predicate for each non-atomic subformula and employing the $\leftrightarrow$ constraint. Inspired by~\cite{PV03,RozierV11}, we define in this section several optimizations to simplify such encoding thus benefiting symbolic DFA construction. These variations indicating different optimizations are combinations of three independent components: (1) the Normal Form~(choose between \BNF or \NNF); (2) the Constraint Form~(choose between \emph{Fussy} or \emph{Sloppy}); (3)the Variable Form~(choose between \emph{Full} or \emph{Lean}). 
In each component one can choose either of two options to make. 
Thus for example, the variation described in Section~\ref{sec:mso} is \BNF-\emph{Fussy}-\emph{Full}. 
Note that \BNF-\emph{Sloppy} are incompatible, as described below, and so there are $2^3-2=6$ viable combinations of the three components above. We next describe the variations in details. 

\vspace{0.2cm}
\noindent{\textbf{\emph{Constraint Form}}}
We call the translation described in Section~\ref{sec:mso} the \emph{Fussy} variation, in which we translate $\phi$ to \MSO formula $\mathsf{mso}(\phi)$ by employing an $\emph{iff}$ constraint (see Section~\ref{sec:mso}). For example:
\begin{equation}\label{eq:example}
\mathsf{t}(\theta_i,x) = (Q_{\theta_i}(x) \leftrightarrow (Q_{\theta_j}(x) \wedge Q_{\theta_k}(x))) \text{ if } \theta_i = (\theta_j \wedge \theta_k)
\end{equation}

We now introduce \emph{Sloppy} variation, inspired by~\cite{RozierV11}, which allows less tight constraints that still hold correctness guarantees thus may speed up the symbolic DFA construction. To better reason the incompatible combination \BNF-\emph{Sloppy}, we specify the description for different normal forms, \NNF and \BNF separately.

For \LTLf formulas in \NNF, the \emph{Sloppy} variation requires only a single implication constraint $\rightarrow$. Specifically the \emph{Sloppy} variation $\mathsf{mso}_s(\phi)$ for \NNF returns \MSO formula $(\exists Q_{\theta_1}) \cdots (\exists Q_{\theta_m})$
$(Q_{\phi}(0) \wedge (\forall x)(\bigwedge_{i=1}^m \mathsf{t}_s(\theta_i,x)))$, where $\mathsf{t}_s(\theta_i)$ is defined just like $\mathsf{t}(\theta_i)$, replacing the $\leftrightarrow$ by $\rightarrow$.
For example translation (\ref{eq:example}) under the \emph{Sloppy} translation for \NNF is $\mathsf{t}_s(\theta_i,x) = (Q_{\theta_i}(x) \rightarrow (Q_{\theta_j}(x) \wedge Q_{\theta_k}(x)))$. 

The \emph{Sloppy} variation cannot be applied to \LTLf formulas in \BNF since the $\leftrightarrow$ constraint defined in function $\mathsf{t}(\theta_i)$ is needed only to handle negation correctly. \BNF requires a general handling of negation.
For \LTLf formulas in \NNF, negation is applied only to atomic formulas such that handled implicitly by the base case $\rho, x \models p \leftrightarrow \rho, x \nvDash \neg p$. Therefore, translating \LTLf formulas in \NNF does not require the $\leftrightarrow$ constraint.
For example, consider \LTLf formula $\phi = \neg Fa$~(in \BNF), where $a$ is an atom. The corresponding \BNF-\textit{Sloppy} variation gives \MSO formula $(\exists Q_{\neg Fa})(\exists Q_{Fa})(Q_{\neg Fa}(0) \wedge ((\forall x)((Q_{\neg Fa}(x) \rightarrow \neg Q_{Fa}(x)) \wedge (Q_{Fa}(x) \rightarrow (Q_a(x) \vee ((x \neq last) \wedge Q_{Fa}(x+1)))))))$ via $\mathsf{mso}_s(\phi)$. Consider finite trace $\rho = (a=0),(a=1)$, $\rho \models \phi$ iff $\rho \models \mathsf{mso}_s(\phi)$ does not hold since $\rho \nvDash \neg Fa$. This happens because $\neg Fa$ requires $(Q_{\neg Fa}(x) \leftrightarrow \neg Q_{Fa}(x))$ as $Fa$ is an non-atomic subformula. Therefore, \emph{Sloppy} variation can only be applied to \LTLf formulas in \NNF.

The following theorem asserts the correctness of the \emph{Sloppy} variation. 

\begin{theorem}\label{thm:msosloppy}
Let $\phi$ be an \LTLf formula in \NNF and $\rho$ be a finite trance. Then $\rho \models \phi$ iff $\mathcal{I}_{\rho} \models \mathsf{mso}_s(\phi)$.
\end{theorem}
The proof here is analogous to that of Theorem~\ref{thm:msofussy}. The crux here is that the $\leftrightarrow$ in $\mathsf{t}(\theta_i)$ is needed only to handle negation correctly. \emph{Sloppy} encoding, however, is applied only to \LTLf formulas in \NNF, so negation can be applied only to atomic propositions, which is handled by the base case $(\neg Q_p(x))$.

\vspace{0.2cm}
\noindent{\textbf{\emph{Variable Form}}}
In all the variations of the \MSO encoding we can get above, we introduced a monadic predicate for each non-atomic subformula in $cl(\phi)$, this is the \emph{Full} variation. We now introduce \emph{Lean} variation, a new variable form, aiming at decreasing the number of quantified monadic predicates. Fewer quantifiers on monadic predicates could benefit symbolic DFA construction a lot since quantifier elimination in $\mathsf{MONA}$ takes heavy cost. 
The key idea of \emph{Lean} variation is introducing monadic predicates only for atomic subformulas and non-atomic subformulas of the form $\phi_j U \theta_k$ or $\phi_j R \theta_k$~(named as $U$- or $R$-subformula respectively).

For non-atomic subformulas that are not $U$- or $R$- subformulas, we can construct \emph{second-order terms} using already defined monadic predicates to capture the semantics of them. Function $\mathsf{lean}(\theta_i)$ is defined to get such second-order terms. Intuitively speaking, $\mathsf{lean}(\theta_i)$ indicates the same positions where $\theta_i$ is true as $Q_{\theta_i}$ does, instead of having $Q_{\theta_i}$ explicitly. We use built-in second-order operators in $\mathsf{MONA}$ to simplify the definition of $\mathsf{lean}(\theta_i)$. $\mathsf{ALIVE}$ is defined using built-in procedures in $\mathsf{MONA}$ to collect all instances along the finite trace. $\mathsf{MONA}$ also allows to apply set union, intersection, and difference for second-order terms, as well as the $-1$ operation (which shifts a monadic predicate backwards by one position). $\mathsf{lean}(\theta_i)$ is defined over the structure of $\theta_i$ as following:
\begin{itemize}
\item If $\theta_i = (\neg\theta_j)$, then $\mathsf{lean}(\theta_i) = (\mathsf{ALIVE} \backslash \mathsf{lean}(\theta_j))$
\item If $\theta_i = (\theta_j \wedge \theta_k)$, then $\mathsf{lean}(\theta_i) = (\mathsf{lean}(\theta_j) \mbox{ inter } \mathsf{lean}(\theta_k))$
\item If $\theta_i = (\theta_j \vee \theta_k)$, then $\mathsf{lean}(\theta_i) = (\mathsf{lean}(\theta_j) \mbox{ union } \mathsf{lean}(\theta_k))$
\item If $\theta_i = (X\theta_j)$, then $\mathsf{lean}(\theta_i) = ((\mathsf{lean}(\theta_j)-1) \backslash \{last\})$
\item If $\theta_i = (N\theta_j)$, then $\mathsf{lean}(\theta_i) = ((\mathsf{lean}(\theta_j)-1) \mbox{ union } \{last\})$
\item If $\theta_i = (\theta_j U \theta_k)$ or $\theta_i = (\theta_j R \theta_k)$, then $\mathsf{lean}(\theta_i) = Q_{\theta_a}$, where $Q_{\theta_a}$ is the corresponding monadic predicate.
\end{itemize}

The following lemma ensures that $\mathsf{lean}(\theta_i)$ keeps the interpretation of each non-atomic subformula $\theta_i \in cl(\phi)$. 

\begin{lemma}\label{lemma:lean}
Let $\phi$ be an \LTLf formula, $\rho$ be a finite trace. Then $\rho,x \models \theta_i$ iff $\mathsf{lean}(\theta_i)(x)$ holds, where $x$ is the position in $\rho$.
\end{lemma}

\begin{proof}
Suppose first that $\rho,x \models \theta_i$. We prove this inductively on the structure of $\theta_i$.
\begin{itemize}
\item If $\theta_i = \neg \theta_j$, then $\mathsf{lean}(\theta_i) = (\mathsf{ALIVE} \backslash \mathsf{lean}(\theta_j))$. $\mathsf{lean}(\theta_i)(x)$ holds since $\mathsf{lean}(\theta_i) = \{x~:~x \notin \mathsf{lean}(\theta_j)\}$ and $\mathsf{lean}(\theta_j) = \{x~:~\rho,x \models \theta_j\}$.

\item If $\theta_i = \theta_j \wedge \theta_k$, then $\mathsf{lean}(\theta_i) = (\mathsf{lean}(\theta_j)~\mbox{inter}~\mathsf{lean}(\theta_k))$. $\mathsf{lean}(\theta_i)(x)$ holds since $\mathsf{lean}(\theta_i) = \{x~:~x \in \mathsf{lean}(\theta_j) \mbox{ and } x \in \mathsf{lean}(\theta_k)\}$, $\mathsf{lean}(\theta_j) = \{x~:~\rho,x \models \theta_j\}$ and $\mathsf{lean}(\theta_k) = \{x~:~\rho,x \models \theta_k\}$.

\item If $\theta_i = \theta_j \vee \theta_k$, then $\mathsf{lean}(\theta_i) = (\mathsf{lean}(\theta_j)~\mbox{union}~\mathsf{lean}(\theta_k))$. $\mathsf{lean}(\theta_i)(x)$ holds since $\mathsf{lean}(\theta_i) = \{x~:~x \in \mathsf{lean}(\theta_j) \mbox{ or } x \in \mathsf{lean}(\theta_k)\}$, $\mathsf{lean}(\theta_j) = \{x~:~\rho,x \models \theta_j\}$ and $\mathsf{lean}(\theta_k) = \{x~:~\rho,x \models \theta_k\}$.

\item If $\theta_i = X\theta_j$, then $\mathsf{lean}(\theta_i) = ((\mathsf{lean}(\theta_j)-1) \backslash \{last\})$. $\mathsf{lean}(\theta_i)(x)$ holds since $\mathsf{lean}(\theta_i) = \{x~:~x \neq last \mbox{ and } x+1 \in \mathsf{lean}(\theta_j)\}$, $\mathsf{lean}(\theta_j) = \{x~:~\rho,x \models \theta_j\}$.

\item If $\theta_i = N\theta_j$, then $\mathsf{lean}(\theta_i) = ((\mathsf{lean}(\theta_j)-1)~\mbox{union}~\{last\})$. $\mathsf{lean}(\theta_i)(x)$ holds since $\mathsf{lean}(\theta_i) = \{x~:~x = last \mbox{ or } x+1 \in \mathsf{lean}(\theta_j)\}$, $\mathsf{lean}(\theta_j) = \{x~:~\rho,x \models \theta_j\}$.

\item If $\theta_i = \theta_j U \theta_k$ or $\theta_i = \theta_j R \theta_k$, then $\mathsf{lean}(\theta_i) = Q_{\theta_a}$. $\mathsf{lean}(\theta_i)(x)$ holds since $\mathsf{lean}(\theta_i)  = Q_{\theta_a} = \{x~:~\rho,x \models \theta_a\}$, where $Q_{\theta_a}$ is the corresponding second-order predicate for formula $\theta_i$.

\end{itemize}

Assume now that $\I_{\rho}^{mso} \models \mathsf{lean}(\theta_i)(x)$ with given interpretations of second-order predicates. We now prove $\rho,x \models \theta_i$ by induction over the structure on $\theta_i$.
\begin{itemize}
\item If $\theta_i = \neg \theta_j$, then $\mathsf{lean}(\theta_i) = (\mathsf{ALIVE} \backslash \mathsf{lean}(\theta_j))$. Since $\mathsf{lean}(\theta_i)(x)$ holds, we also have that $x \in \mathsf{lean}(\theta_i)$ iff $x \notin \mathsf{lean}(\theta_j)$. It follows by induction that $\rho,x \models \theta_i$.

\item If $\theta_i = \theta_j \wedge \theta_k$, then $\mathsf{lean}(\theta_i) = (\mathsf{lean}(\theta_j) ~\mbox{inter}~ \mathsf{lean}(\theta_k))$. Since $\mathsf{lean}(\theta_i)(x)$ holds, we also have that $x \in \mathsf{lean}(\theta_i)$ iff $x \in \mathsf{lean}(\theta_j)$ and $x \in \mathsf{lean}(\theta_k)$. It follows by induction that $\rho,x \models \theta_i$.

\item If $\theta_i = \theta_j \vee \theta_k$, then $\mathsf{lean}(\theta_i) = (\mathsf{lean}(\theta_j) ~\mbox{union}~ \mathsf{lean}(\theta_k))$. Since $\mathsf{lean}(\theta_i)(x)$ holds, we also have that $x \in \mathsf{lean}(\theta_i)$ iff $x \in \mathsf{lean}(\theta_j)$ or $x \in \mathsf{lean}(\theta_k)$. It follows by induction that $\rho,x \models \theta_i$.

\item If $\theta_i = X\theta_j$, then $\mathsf{lean}(\theta_i) = ((\mathsf{lean}(\theta_j)-1) \backslash \{last\})$. Since $\mathsf{lean}(\theta_i)(x)$ holds, we also have that $x \neq last \mbox{ and } x+1 \in \mathsf{lean}(\theta_j)$. It follows by induction that $\rho,x \models \theta_i$.

\item If $\theta_i = N\theta_j$, then $\mathsf{lean}(\theta_i) = ((\mathsf{lean}(\theta_j)-1) ~\mbox{union}~ \{last\})$. Since $\mathsf{lean}(\theta_i)(x)$ holds, we also have that $x = last \mbox{ or } x+1 \in \mathsf{lean}(\theta_j)$. It follows by induction that $\rho,x \models \theta_i$.

\item If $\theta_i = \theta_j U \theta_k$ or $\theta_i = \theta_j R \theta_k$, then $\mathsf{lean}(\theta_i) = Q_{\theta_a}$, where $Q_{\theta_a}$ is the corresponding second-order predicate. It follows by induction that $\rho,x \models \theta_i$.
\end{itemize}

\end{proof}

Finally, we define \emph{Lean} variation based on function $\mathsf{lean}(\phi)$. \emph{Lean} variation $\mathsf{mso}_\lambda(\phi)$ returns \MSO formula $(\exists Q_{\theta_1})\ldots(\exists Q_{\theta_n})$ 
 $(\mathsf{lean}(\phi)(0) \wedge ((\forall x)(\bigwedge_{a=1}^{n}\mathsf{t}_\lambda(\theta_a,x))))$, where $n$ is the number of $U$- and $R$- subformulas $\theta_a \in cl(\phi)$, and $\mathsf{t}_\lambda(\theta_a,x)$ is defined as follows: if $\theta_a = (\theta_j U \theta_k)$, then $\mathsf{t}_\lambda(\theta_a,x) = (Q_{\theta_a}(x) \leftrightarrow (\mathsf{lean}(\theta_k)(x) \vee ((x \neq last) \wedge \mathsf{lean}(\theta_j)(x) \wedge Q_{\theta_a}(x+1))))$;
if $\theta_a = (\theta_j R \theta_k$), then $\mathsf{t}_\lambda(\theta_a,x) = (Q_{\theta_a}(x) \leftrightarrow (\mathsf{lean}(\theta_k)(x) \wedge ((x = last) \vee \mathsf{lean}(\theta_j)(x) \vee Q_{\theta_a}(x+1))))$.
The following theorem guarantees the correctness of \textit{Lean} variation. 
\begin{theorem}\label{thm:lean}
Let $\phi$ be an \LTLf formula, $\rho$ be a finite trace. Then $\rho \models \phi$ iff $\mathcal{I}_{\rho} \models \mathsf{mso}_\lambda(\phi)$.
\end{theorem}

\begin{proof}
If $\phi$ is a propositional atom $p$, then $\mathsf{mso}_\lambda(\phi)=Q_p(0)$. It is true that $\rho \models \phi$ iff $\I_\rho \models \mathsf{mso}_\lambda(\phi)$. If $\phi$ is an nonatomic formula, we prove this theorem in two directions.

Suppose first that $\rho$ satisfies $\phi$. We expand the monadic structure $\I_\rho$ with interpretations for $Q_{\theta_1}, Q_{\theta_2},\ldots,Q_{\theta_n}$ by setting $Q_{\theta_a} = \{x~:~\rho,x \models \theta_a\}$. Let the expanded structure be $\I_\rho^{mso}$. By assumption, $\mathsf{lean}(\phi)(0)$ holds in $\I_\rho^{mso}$.
It remains to prove that $\I_\rho^{mso}\models(\forall x)(\bigwedge_{a=1}^{n}\mathsf{t}_\lambda(\theta_a,x))$, for each $U$ or $R$ subformula $\theta_a\in cl(\phi)$.
\begin{itemize}

\item If $\theta_a = (\theta_j U \theta_k)$, then $\mathsf{t}_\lambda(\theta_a,x) = (\mathsf{lean}(\theta_a)(x) \leftrightarrow (\mathsf{lean}(\theta_k)(x) \vee ((x \neq last) \wedge \mathsf{lean}(\theta_j)(x) \wedge \mathsf{lean}(\theta_a)(x+1))))$. This holds, since $\mathsf{lean}((\theta_j U \theta_k)) = \{x~:~\rho,x \models \theta_j U \theta_k\} = \{x~:~\rho,x \models \theta_k \mbox{ or } x \neq last \mbox{ with } \rho,x \models \theta_j \mbox{ also } \rho,x+1 \models \theta_a \}$, $\mathsf{lean}(\theta_j) = \{x~:~\rho,x \models \theta_j\}$, and $\mathsf{lean}(\theta_k) = \{x~:~\rho,x \models \theta_k\}$ with Lemma~\ref{lemma:lean};

\item If $\theta_a = (\theta_j R \theta_k)$, then $\mathsf{t}_\lambda(\theta_a,x) = (\mathsf{lean}(\theta_a)(x) \leftrightarrow (\mathsf{lean}(\theta_k)(x) \wedge ((x = last) \vee \mathsf{lean}(\theta_j)(x) \vee \mathsf{lean}(\theta_a)(x+1))))$. This holds, since $\mathsf{lean}((\theta_j R \theta_k)) = \{x~:~\rho,x \models \theta_j R \theta_k\} = \{x~:~\rho,x \models \theta_k \mbox{ with } x = last \mbox{ or } \rho,x \models \theta_j \mbox{ or } \rho,x+1 \models \theta_a \}$, $\mathsf{lean}(\theta_j) = \{x~:~\rho,x \models \theta_j\}$, and $\mathsf{lean}(\theta_k) = \{x~:~\rho,x \models \theta_k\}$ with Lemma~\ref{lemma:lean}.
\end{itemize}

Assume now that $\I_\rho\models \mathsf{mso}_\lambda(\phi)$. This means that there is an expansion of $\I_\rho$ with monadic interpretations $Q_{\theta_a}$ for each element $\theta_a$ of U or R subformulas in $cl(\phi)$ such that this expanded structure $\I_\rho^{mso}\models (\mathsf{lean}(\phi)(0))\wedge ((\forall x)(\bigwedge_{a=1}^{n} \mathsf{t}_\lambda(\theta_a,x)))$. If $\phi$ is not an $R$ or $U$ subformula, then it has been proven by Lemma~\ref{lemma:lean} that if $x \in \mathsf{lean}(\phi)$, then $\rho,x\models\phi$. We now prove by induction on $\phi$ that if $x \in Q_{\theta_a}$, then $\rho,x\models\phi$. Since $\I_\rho^{mso}\models (\mathsf{lean}(\phi)(0))$, it follows that $\rho,0\models\phi$. 
\begin{itemize}
\item
If $\phi= (\theta_j U \theta_k$), then $\mathsf{t}_\lambda(\phi,x) = (\mathsf{lean}(\phi)(x) \leftrightarrow (\mathsf{lean}(\theta_k)(x) \vee ((x \neq last) \wedge \mathsf{lean}(\theta_j)(x) \wedge \mathsf{lean}(\phi)(x+1))))$. Since $\mathsf{t}_\lambda(\phi)$ holds at every point $x$ of $\I_\rho^{mso}$, it follows that $x\in \mathsf{lean}(\phi)$ iff $x \in \mathsf{lean}(\theta_k)$ or $x \neq last \mbox{ with } x \in \mathsf{lean}(\theta_j) \mbox{ also }x+1\in \mathsf{lean}(\phi)$. Moreover, $\mathsf{lean}(\phi) = Q_{\theta_a}$, where $Q_{\theta_a}$ is the corresponding second-order predicate. Thus, by induction hypothesis $\rho,x\models\phi$.
\item
If $\phi= (\theta_j R \theta_k$), then $\mathsf{t}_\lambda(\phi,x) = (\mathsf{lean}(\phi)(x) \leftrightarrow (\mathsf{lean}(\theta_k)(x) \wedge ((x = last) \vee \mathsf{lean}(\theta_j)(x) \vee \mathsf{lean}(\phi)(x+1))))$. Since $\mathsf{t}_\lambda(\phi)$ holds at every point $x$ of $\I_\rho^{mso}$, it follows that $x\in \mathsf{lean}(\phi)$ iff $x \in \mathsf{lean}(\theta_k)$ with $x = last \mbox{ or } x \in \mathsf{lean}(\theta_j) \mbox{ or }x+1\in \mathsf{lean}(\phi)$. Moreover, $\mathsf{lean}(\phi) = Q_{\theta_a}$, where $Q_{\theta_a}$ is the corresponding second-order predicate.  Thus, by induction hypothesis $\rho,x\models\phi$.

\end{itemize}
\end{proof}

Having defined different variations of the \MSO encoding, we now provide variations of the Compact \MSO encoding described in Section~\ref{sec:omso}. 

\vspace{0.2cm}
\noindent\textbf{\emph{Sloppy Formulation}}
The formulation described in Section~\ref{sec:omso} strictly tracks the computation over each BDD $B_q$, which we refer to \emph{Fussy} formulation. That is, for each nonterminal node $\alpha$, both the forward computation and previous computation must be tracked. This causes a high logical complexity in the formulation. An alteration to diminish the logical complexity is to utilize a \emph{Sloppy Formulation}, analogous to the \emph{Sloppy} variation described above, that only tracks the forward computation. Since the previous computations are not tracked, none of the computations leading to terminal node $0$ of the BDD $B_f'$ enable an accepting condition. 

To define the accepting condition of \emph{Sloppy Formulation}, we have
\begin{align*}
\mathsf{Racc_s} = \left(\bigvee_{(\beta,v,d) \in \mathsf{PreT}(B_f',0)}[x \in N_\beta \wedge x \in^d v] \right) \rightarrow (x \neq 0).
\end{align*}
Moreover, $\mathsf{node_s}$ only requires $\mathsf{PostCon}$ of $\mathsf{node}$. Therefore, we have
\begin{align*}
\mathsf{node_s} = \bigwedge_{1\leq \alpha\leq u} \mathsf{PostCon}.
\end{align*}

The second-order formula $\mathsf{Rev_s}(\F_\psi)$ of \emph{Sloppy Formulation} is defined as following:\\
\centerline{
$\mathsf{Rev_s}(\F_\psi) = (\exists V_0) \ldots (\exists V_{k-1})(\exists N_1) \ldots (\exists N_u)(\forall x)(\mathsf{Rinit} \wedge \mathsf{node_s}$}
\centerline{
$\wedge \mathsf{Rterminal} \wedge \mathsf{Racc_s} \wedge \mathsf{roots})$,}
where $\mathsf{Rinit}, \mathsf{Rterminal}$ and $\mathsf{roots}$ are defined as in Section~\ref{sec:omso}. Therefore, let $\mathsf{Cmso_s}(\phi)$ be the \emph{Sloppy Formulation} of the Compact \MSO encoding, we define $\mathsf{Cmso_s}(\phi) = \mathsf{Rev_s}(\F_\psi)$ asserting the truth of $\phi$ at position 0, where $\psi$ is the corresponding \PLTLf formula of $\phi$, and $\F_\psi$ is the symbolic DFA of $\psi$. The following theorem asserts the correctness of the \emph{Sloppy Formulation}.

\begin{theorem}\label{thm:cmsosloppy}
The models of formula $\mathsf{Cmso_s}(\phi)$ are exactly the words satisfying $\phi$.
\end{theorem}
The proof here is analogous to that of the \emph{Fussy Formulation}, where the crux is that we define the computation trace on a BDD as a sequence of sets of BDD nodes, instead of just a specific sequence of BDD nodes, see the definition of $\mathsf{node_s}$. Such definition still keeps unambiguous formulation of the symbolic DFA since we have stronger constraints on the accepting condition, as shown in the definition of $\mathsf{Racc_s}$.

\subsection{Experimental Results}
Having presented different optimizations, we now have 6 variations of the \MSO encoding corresponding to specific optimizations, which are \BNF-\emph{Fussy-Full}, \BNF-\emph{Fussy-Lean}, \NNF-\emph{Fussy-Full}, \NNF-\emph{Fussy-Full}, \NNF-\emph{Sloppy-Full} and \NNF-\emph{Sloppy-Lean}. Moreover, we have two variations of the Compact \MSO encoding, which are \emph{Fussy} and \emph{Sloppy}. 
The experiments were divided into two parts and resulted in two major findings. First we explored the benefits of the various optimizations of \MSO encoding and showed that the most effective one is that of \emph{Lean}. Second, we aimed to answer the question whether second-order outperforms first-order in the context of \LTLf-to-automata translation. To do so, we compared the best performing \MSO encoding and Compact \MSO encoding against the \FOL encoding and showed the superiority of first-order.

\noindent\textbf{\emph{Correctness}}
The correctness of the implementation of different encodings was evaluated by comparing the DFAs in terms of the number of states and transitions generated from each encoding. No inconsistencies were discovered.

\begin{figure}[t]
  \centering
  \includegraphics[width=0.65\linewidth]{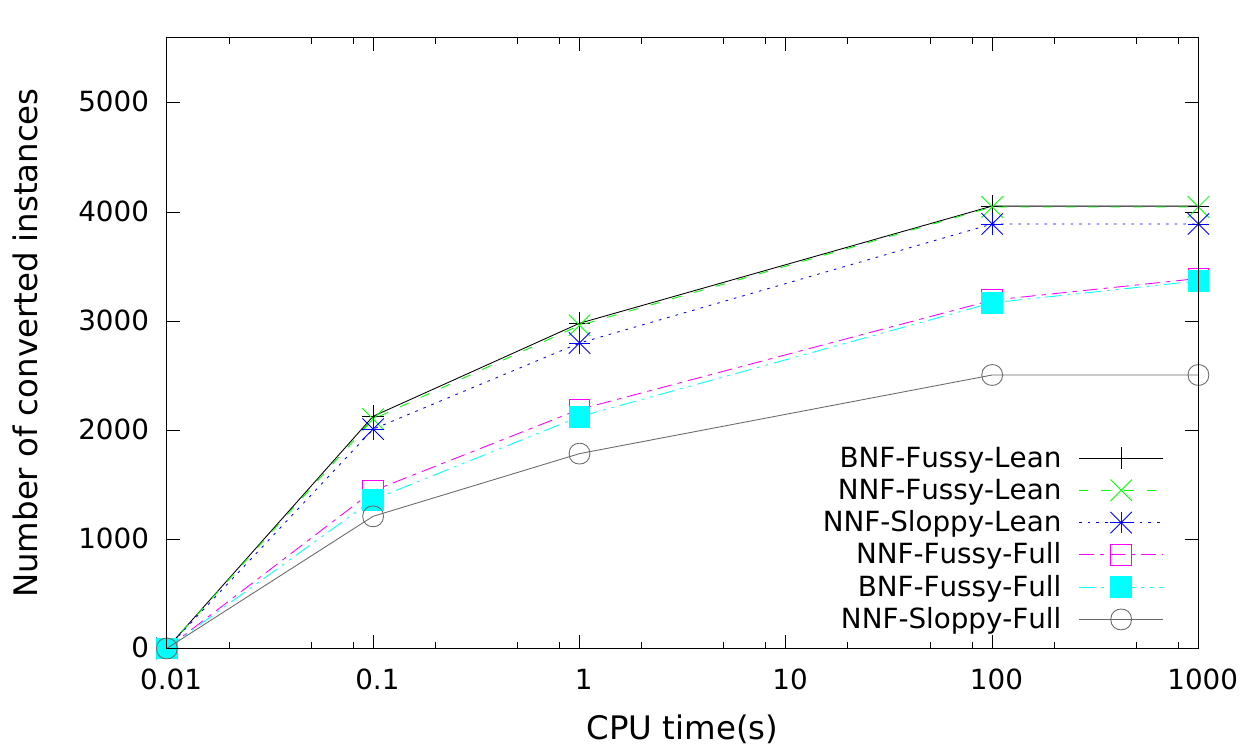}\\
  \caption[small]{Comparison over 6 variations of \MSO encoding}\label{fig:fullvslean}
  
\end{figure}

\begin{figure}[t]
  \centering
  \includegraphics[width=0.65\linewidth]{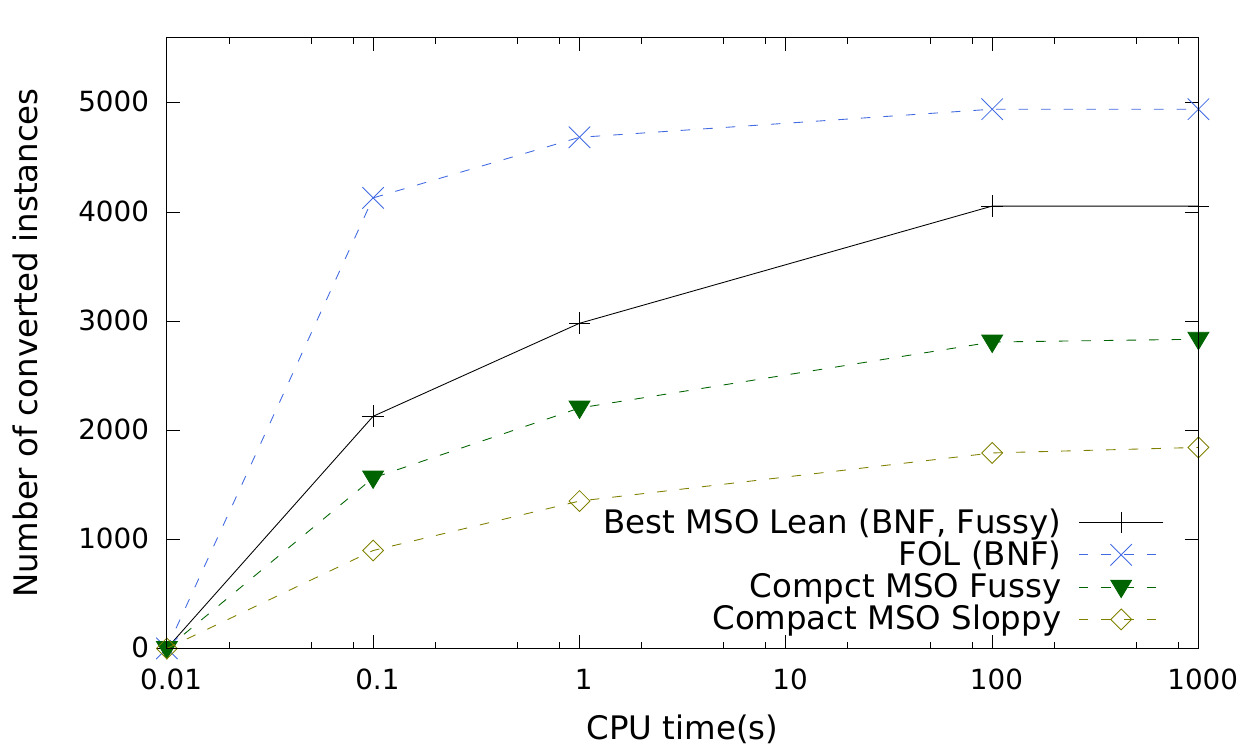}\\
  \caption[small]{Overall comparison of \FOL, \MSO and Compact \MSO encodings}\label{fig:final}

\end{figure}

\noindent\textbf{\emph{Lean constraint form is more effective in \MSO encodings.}}
Figure~\ref{fig:fullvslean} presents the number of converted instances of each variation of \MSO encoding, where the upper three are all for \emph{Lean} variations and the lower ones are for \emph{Full} variations. The choice of \BNF vs \NNF did not have a major impact, and neither did the choice of \emph{Fussy} vs \emph{Sloppy}.
The one optimization that was particularly effective was that of \emph{Lean} variation.  The best-performing \MSO encoding was \BNF-\emph{Fussy}-\emph{Lean}. While in the Compact \MSO encoding, the \emph{Fussy} variation highly outperforms that of \emph{Sloppy}, as shown in Figure~\ref{fig:final}.

\noindent\textbf{\emph{First-order logic dominates second-order logic for \LTLf-to-automata translation.}}
As presented in Figure~\ref{fig:final}, \FOL encoding shows its superiority over second-order encodings performance-wise, which are \MSO encoding and Compact \MSO encoding. Thus, the use of second-order logic, even under sophisticated optimization, did not prove its value in terms of performance. This suggests that nevertheless second-order encoding indicates a much simpler quantificational structure which theoretically leads to more potential space to optimize, it would be useful to have first-order as a better way in the context of \LTLf-to-automata translation in practice. 


\section{Concluding Remarks}\label{sec:con}
In this paper, we revisited the translation from \LTLf to automata and presented new second-order encodings, \MSO encoding and Compact \MSO encoding with various optimizations. Instead of the syntax-driven translation in \FOL encoding, \MSO encoding provides a semantics-driven translation. Moreover, \MSO encoding allows a significantly simpler quantificational structure, which requires only a block of existential second-order quantifiers, followed by a single universal first-order quantifier, while \FOL encoding involves an arbitrary alternation of quantifiers. The Compact \MSO encoding simplifies further the syntax of the encoding, by introducing more second-order variables. Nevertheless, empirical evaluation showed that first-order encoding, in general, outperforms the second-order encodings. This finding suggests first-order encoding as a better way for \LTLf-to-automata translation.

To obtain a better understanding of the performance of second-order encoding of \LTLf, we looked more into \MONA. An interesting observation is that $\mathsf{MONA}$ is an ``aggressive minimizer'': after each quantifier elimination, \MONA re-minimizes the DFA under construction. Thus, the fact that the second-order encoding starts with a block of existential second-order quantifiers offers no computational advantage, as $\mathsf{MONA}$ eliminates the second-order quantifiers one by one, performing computationally heavy minimization after each quantifier. Therefore, a possible improvement to \MONA would enable it to eliminate a whole \emph{block} of quantifiers of the same type (existential or universal) in one operation, involving only one minimization. Currently, the quantifier-elimination strategy of one quantifier at a time is deeply hardwired in \MONA, so the suggested improvement would require a major rewrite of the tool. We conjecture that, with such an extension of \MONA, the second-order encodings would have a better performance, but this is left to future work.

Beyond the unrealized possibility of performance gained via second-order encodings, another motivation for studying such encodings is their greater expressivity. The fact that \LTLf is equivalent to \FOL~\cite{Kam68} shows limited expressiveness of \LTLf. For this reason it is advocated in~\cite{DV13} to use \emph{Linear Dynamic Logic}~(\LDLf) to specify ongoing behavior. \LDLf is expressively equivalent to \MSO, which is more expressive than \FOL. Thus, automata-theoretic reasoning for \LDLf, for example, reactive synthesis~\cite{DegVa15}, cannot be done via first-order encoding and requires second-order encoding. Similarly, synthesis of \LTLf with incomplete information requires the usage of second-order encoding \cite{GiacomoV16}. We leave this too to future research.



\vspace{0.1cm}
\noindent\textbf{Acknowledgments.}
Work supported in part by China HGJ Project No.~2017ZX01038102-002, NSFC Projects No.~61572197, No.~61632005 and No.~61532019, NSF grants IIS-1527668, IIS-1830549, and by NSF Expeditions in Computing project ``ExCAPE: Expeditions in Computer Augmented Program Engineering". Special thanks to Jeffrey M. Dudek and Dror Fried for useful discussions.

\bibliographystyle{splncs04}
\bibliography{cav}

\end{document}